\documentclass{article}

\usepackage[noadjust]{cite}

\usepackage{scrextend}
\usepackage[margin=1.2in]{geometry}
\usepackage{amsmath,amssymb,amsfonts,graphicx,amsthm,nicefrac,mathtools,bm}
\usepackage[hidelinks]{hyperref}
\usepackage[english]{babel}
\usepackage{times}
\usepackage[T1]{fontenc}
\usepackage{bbm}
\usepackage{color,xcolor}
\usepackage{xspace}
\usepackage{rotating,multirow}
\usepackage{comment}
\usepackage[ruled,vlined]{algorithm2e}
\usepackage{algorithmicx}
\usepackage{booktabs}

\newtheorem{theorem}{Theorem}

\newtheorem{lemma}{Lemma}

\theoremstyle{definition}

\theoremstyle{remark}

\makeatletter
\newtheorem*{rep@theorem}{\rep@title}
\newcommand{\newreptheorem}[2]{%
\newenvironment{rep#1}[1]{%
 \def\rep@title{#2 \ref{##1}}%
 \begin{rep@theorem}}%
 {\end{rep@theorem}}}
\makeatother
\newreptheorem{theorem}{Theorem}
\newreptheorem{lemma}{Lemma}
\newreptheorem{corollary}{Corollary}
\newreptheorem{proposition}{Proposition}

\newcommand{\PP}[1]{\mathbb{P}\left\{{#1}\right\}} % Probability
\newcommand{\ignore}[1]{}
\newcommand{\ejc}[1]{\textcolor{red}{[EJC: #1]}}

\newcommand{\rfb}[1]{\textcolor{teal}{[RFB: #1]}}

\newcommand{\alphalo}{\alpha_{\text{\rm lo}}}
\newcommand{\alphahi}{\alpha_{\text{\rm hi}}}

\def\PP{{\mathbb P}}
\def\RR{{\mathbb R}}
\def\NN{{\mathbb N}}

\mathtoolsset{showonlyrefs}

\usepackage{tikz}
\usetikzlibrary{calc}

\author{Yaniv Romano\thanks{Department of Statistics, Stanford University} \quad Rina Foygel Barber\thanks{Department of Statistics, University of Chicago} \quad Chiara Sabatti\footnotemark[1] \ \thanks{Department of Biomedical Data Science, Stanford University} \quad  
Emmanuel J.~Cand{\`e}s\footnotemark[1] \ \thanks{Department of Mathematics, Stanford University} }
\title{With Malice Towards None:\\
Assessing Uncertainty via Equalized Coverage}

\date{August, 2019}

\begin{document}

\maketitle

\begin{abstract}

An important factor to guarantee a fair use of data-driven recommendation systems is that we should be able to communicate their uncertainty to decision makers.
This can be accomplished by constructing prediction intervals, which provide an  intuitive measure of the limits of predictive performance.
To support equitable treatment, we force the construction of such intervals to be unbiased in the sense that their coverage must be equal across all protected groups of interest. We present an operational methodology that achieves this goal by offering rigorous distribution-free coverage guarantees holding in finite samples. 
Our methodology, {\em equalized coverage}, is flexible as it can be viewed as a wrapper around any predictive algorithm. We test the applicability of the proposed framework on real data, demonstrating that equalized coverage constructs unbiased prediction intervals, unlike competitive methods.

\end{abstract}

\section{Introduction}

\subsection{The problem of equitable treatment}

We are increasingly turning to machine learning systems to support
human decisions. While decision makers may be subject to
many forms of prejudice and bias, the promise and hope is that
machines would be able to make more equitable decisions. Unfortunately, whether because they are fitted 
%\rfb{fitted?}\yr{fixed}
on already biased data or otherwise, there are concerns that some of these data driven recommendation systems treat members of different classes differently, perpetrating biases, providing different degrees of utilities, and inducing disparities. The  examples that have emerged are quite varied:
%\footnote{The notion of equalized coverage, advocated in this paper, is useful regardless of whether the initial disparity is due to factors of inequality/bias, or instead due to genetic risk.}
\begin{enumerate}
	\item \label{item:criminal_justice} \textbf{Criminal justice}: courts in the United States use COMPAS---a commercially available algorithm to assess a criminal defendant's likelihood of becoming a recidivist---to help them decide who should receive parole, based on records collected through the criminal justice system. In 2016 ProPublica analyzed COMPAS and ``found that black defendants were far more likely than white defendants to be incorrectly judged to be at a higher risk of recidivism, while white defendants were more likely than black defendants to be incorrectly flagged as low risk'' \cite{dieterich2016compas}.\footnote{\href{https://www.propublica.org/article/machine-bias-risk-assessments-in-criminal-sentencing}{https://www.propublica.org/article/machine-bias-risk-assessments-in-criminal-sentencing}}
	\item \label{item:face_recognition} \textbf{Recognition system}: the department of motor vehicles (DMV) uses facial recognition tools to detect people with false identities, by comparing driver's license or ID photos with other DMV images on file. In a related context, the authors of \cite{buolamwini2018gender} evaluated the performance of three commercial classification systems that use facial images to predict individuals' gender, and reported that the overall classification accuracy on 
	%\rfb{``of'' should perhaps be ``on'' or ``for''? i.e., ``of'' sounds like we are discussing these individuals' ability to perform prediction}\yr{fixed}
	male individuals was higher than female individuals. They also found that the predictive performance on 
	%\rfb{same}\yr{fixed}
	lighter-skinned individuals was higher than darker-skinned individuals.
	\item \label{item:college_admission} \textbf{College admissions}: a college admission office may be interested in a new algorithm for predicting the college GPA of a candidate student at the end of their Sophomore year, by using features such as high-school GPA, SAT scores, AP courses taken and scores, intended major, levels of physical activity, and so on.  On a similar matter, the work reported in \cite{gardner2019evaluating} studied various data-driven algorithms that aim to predict whether a student will drop out from a massive open online course (MOOC).
	Using a large data set available from \cite{gardner2018morf}, the authors found that in some cases there are noticeable differences between the models' predictive performance on male students compared to female students.

% 	The work reported in \cite{kleinberg2018algorithmic} follows the National Education Longitudinal Study of 1988 and uses a related publicly available data set that contains information about college students and their GPA. Their study suggests that a predictive model fitted on this data tends to predict higher college grades for white students compared to black students. \rfb{Does this mean, higher raw prediction, or higher prediction relative to the student's actual grade?}\yr{I am still trying to figure this out, at the moment I think it will be safer to keep it vague.}
	
	\item \label{item:disease_risk} \textbf{Disease risk}: 
	healthcare providers may be interested in predicting the chance that an individual develops certain disorders.
	Diseases with a genetic component have different frequencies in different human populations, reflecting the fact that disease causing mutations arose at different times and in individuals residing in different areas: for example, Tay-Sachs disease is approximately 100 times more common in infants of Ashkenazi Jewish ancestry (central-eastern Europe) than in non-Jewish infants \cite{kaback1977tai}. 
	The genotyping of DNA polymorphisms can lead to more precise individual risk assessment  than that derived from simply knowing which ethnic group the individual belongs to. However, given our still partial knowledge of the disease causing mutations and their prevalence in different populations, the precision of these estimate varies substantially across ethnic groups. For instance, the study reported in [13] found a preference for European genetic variants over non-European variants in two genomic databases that are widely used by clinical geneticists (this reflects the fact that most studies have been conducted on European populations). Relying on this information only would result in predictions that are more accurate for individuals of European descent than for others. 
	%The Clinical Genome Resource (ClinGen) Ancestry and Diversity Working Group highlights the need to develop guidance on race, ethnicity, and ancestry (REA) in data collection and use in clinical genomics to avoid discriminatory practices and disparities in health delivery \cite{imperative}.
\end{enumerate}
%As we can see from the breadth of these examples, the notion of equalized coverage \color{red} CHIARA: what notion of equalized coverage? I do not think we have defined it yet. I would simply say that ``the breadth of these examples underscores how data must be interpreted with care. Indeed, policymakers have
%issued a call that''\color{black}
%is useful regardless of whether the disparity is due to factors of inequality/bias, or instead due to genetic risk.

%Data must be treated with care and policymakers have
%issued a call that~\cite{WhiteHouse}

The breadth of these examples underscores how data must be interpreted with care; the method that is advocated in this paper is useful regardless of whether the disparity is due to factors of inequality/bias, or instead due to genetic risk. Indeed, policymakers have
issued a call that~\cite{WhiteHouse}
\begin{quote}
	\emph{``we must uphold our fundamental values so these systems are neither
		destructive nor opportunity limiting. [...]  In order to ensure that
		growth in the use of data analytics is matched with equal innovation to
		protect the rights of Americans, it will be important to support
		research into mitigating algorithmic discrimination, building systems
		that support fairness and accountability, and developing strong data
		ethics frameworks.''}
\end{quote}
This is a broad call, that covers multiple aspects of data collection, mining and interpretation; clearly, a response requires a multi-faced approach. 
Encouragingly, the machine learning community is beginning to respond to this challenge.  A major area of study has been to propose mathematical definitions of appropriate notions of fairness
\cite{dwork2012fairness,hardt2016equality,dieterich2016compas,zafar2017fairness,hebert2018multicalibration, kim2018multiaccuracy} or algorithmic models of fairness \cite{Kusner2017}. In many cases, these definitions are an attempt to articulate in mathematical terms what it means not to discriminate on the basis of ``protected characteristics'';
the U.~S.~law identify these as sex, race, age, disability, color, creed, national origin, religion, or genetic information. 
%these characteristic are identified by the U.~S.~law and include sex, race, age, disability, color, creed, national origin, religion, or genetic information. 
Now, discrimination can take many forms, and it is not surprising that it might be difficult to identify one analytical property that detects it in every context. Moreover, the call above is broader than the specific domains where discrimination is forbidden by law and invites us to develop analytical frameworks that guarantee an ethical use of data.

\subsection{Responses from the machine learning community}

To understand the complexity of the problem, it is useful to consider the  task  of predicting the value of $Y$, a binary $\{0,1\}$ variable, with a guess $\hat{Y}$. We assume that $Y=1$ represents a more ``favorable'' state, and that the value of $\hat{Y}$ will influence the decider, so that predicting $\hat{Y}=1$ for some individuals gives them an advantage.
In this context, $\PP(\hat{Y}=0 \mid Y=1)$, the {\em false negative rate}, represents the probability with which an opportunity is denied to a ``well deserving'' individual.
%This is an important error rate to control in the case of a parole \rfb{do we mean this as an example? e.g., could write, ``This is an important error rate to control in scenarios such as deciding parole''}
%(see Example \ref{item:criminal_justice}).
%\rev{
It is obvious that this is a critical error rate to control in scenarios such as deciding parole (see Example \ref{item:criminal_justice}):
%,where the decisions being made pertain to an individual's fundamental right to freedom.
freedom is a fundamental right, and nobody should be deprived of it needlessly.
%} \ejc{I prefer the earlier phrasing.} 
We then wish to require that $\PP(\hat{Y}=0 \mid Y=1, A=a)$ is equal across values of the protected attribute $A$ \cite{hardt2016equality}. In the case of distributions of goods (as when giving a loan), one might argue for parity of other measures such as $\PP(\hat{Y}=1\mid  A=a)$ which would guarantee that resources are distributed equally across the different population categories \cite{dwork2012fairness,feldman2015certifying,zafar2017fairnessconstraints}. Indeed, these observations are at the bases of two notions of fairness considered in the literature.
%\yr{SB and RFB suggested to shorten the next paragraph.}

Researchers have noted several problems with fairness measures that, as the above, ask for (approximate) parity of some statistical measure across all of these groups. Without
%the pretense of
providing
a complete discussion,
%\rfb{this phrasing sounds a bit odd / negative}\yr{replaced ``the pretense of'' with ``providing''}
we list some of these problems here. (a) To begin with, it is usually unclear how to design algorithms that would actually obey these notions of fairness from finite samples, especially in situations where the outcome of interest or protected attribute is continuous. (b) Even if we could somehow
``operationalize'' the fairness program, these
measures are usually incompatible: it is provably
impossible to design an algorithm that obeys 
all notions of
%more than one notion of
fairness simultaneously~\cite{chouldechova2017fair,kleinberg2017inherent}.
(c) This is particularly troublesome, as the appropriate measure appears to be context dependent.
%For example, the true positive and false negative rates might not alwayand their importance in other contexts without making assumptions on how decisions will be made. 
In Example \ref{item:disease_risk} above, where $Y=1$ corresponds to having Tay-Sachs, whose rate differs across populations. Due to the unbalanced nature of the disease, one would expect the predictive model to have a lower true positive rate for non-Jewish infants than that of Ashkenazi Jewish infants (for which the disease is much more common). Here, forcing parity of true positive rates \cite{hardt2016equality} would conflict with accurate predictions for each group \cite{hebert2018multicalibration}. (d)
Finally, and perhaps more importantly, researchers have argued that enforcing frequently discussed fairness criteria ``can often harm the very groups that these measures were designed to protect'' \cite{corbett2018measure}.

%\yradd{For instance, consider the case of Example \ref{item:disease_risk}, where Tay-Sachs disease has different rates across different populations. Due to the unbalanced nature of the problem, one may expect the predictive model to have a lower true positive rate for non-Jewish infants than that of Ashkenazi Jewish infants (for which the disease is much more common). Here, forcing parity of true positive rate \cite{hardt2016equality} would conflict with accurate predictions for each group \cite{hebert2018multicalibration}.}

In light of this, it has been suggested to decouple the statistical problem of risk assessment from the policy problem of taking actions and designing interventions. Quoting from \cite{corbett2018measure}, ``an algorithm might (correctly) infer that a defendant has a 20\% chance of committing a violent crime if released, but that fact does not, in and of itself, determine a course of action.'' Keeping away from policy then, how can we respond to the call in \cite{WhiteHouse} and provide a policymaker the best information that can be learned from data while supporting equitable treatment?
Our belief is that multiple approaches will be needed, and with this short paper our aim is to introduce an additional tool to evaluate the performance of algorithms across different population groups. 

\subsection{This paper: equalized coverage}

\newcommand{\Ylo}{\hat{Y}_{\text{lo}}}
\newcommand{\Yhi}{\hat{Y}_{\text{hi}}}

One fundamental way to support data ethics is not to overstate the power of algorithms and data-based predictions, but rather always accompany these with measures of uncertainty that are easily understandable by the user.
This can be done, for example, by providing a plausible
range of predicted values for the outcome of interest. 
For instance, consider a recommendation system for college admission (Example \ref{item:college_admission}), not knowing about the accuracy of the prediction algorithm, we would like to produce for, each student, a predicted GPA interval $[\Ylo, \ \Yhi]$ obeying the following two properties: the interval should be faithful in the sense that the true \emph{unknown} outcome $Y$ lies within the predicted range 90\% of the time, say; second, this should be unbiased in that {\em the average coverage should be the same within each group}. 

Such a predictive interval has the virtue of informing
the decision maker about the evidence machine learning can provide while being explicit about the limits of predictive performance. If the
interval is long, it just means that the predictive model can say little.
%% And that perhaps more data should be collected.
Each group enjoys identical treatment, receiving \emph{equal
	coverage} (e.g., 90\%, or any level the decision maker wishes to achieve). Hence, the results of data analysis are unbiased to all. In particular, if the larger sample size available for one group overly influences the fit, leading to poor performance in the other groups,  
the prediction interval will 
make this immediately apparent through much wider confidence bands for the groups with fewer samples. Prediction intervals with equalized coverage, then, naturally assess and communicate 
the fact that an algorithm has varied levels of performance on different subgroups. 
One might consider  equal length equalized coverage intervals as a landmark of truly ``fair'' treatment: realizing that the data at hand leads to intervals of different lengths should motivate the user to enlarge the data collection. %, which one might consider   %, which is recognized as an important topic in the literature on algorithmic fairness. 

It seems a priori impossible to present information to the policymaker in
such a compelling fashion without a strong model for dependence of the
response $Y$ on the features $X$ or protected attributes $A$.  In our college admission example, one may have trained a wide array of complicated predictive algorithms such as random forests or deep neural networks, each with its own suite of parameters; for all practical purposes, the fitting procedure may just as well be a black box. The surprise is that such a feat is possible under no assumption other than that of having samples that are drawn exchangeably---e.g., they may be drawn i.i.d.---from a population of interest. We propose a concrete 
%meta-algorithm,
procedure,
which acts as a wrapper around the predictive model, to produce perfectly valid prediction intervals that provably satisfy the equalized coverage constraint for any black box algorithm, sample size and distribution.
Such a
%meta-algorithm can be obtained
procedure can be formulated
by refining tools
from conformal inference, a general methodology for constructing prediction intervals~\cite{vovk1999machine,papadopoulos2002inductive,vovk2005algorithmic,vovk2009line,lei2013distribution,lei2018distribution,romano2019conformalized}. Our contribution is nevertheless a little outside of classical conformal inference as we seek a form of {\em conditional rather than marginal}  coverage guarantee~\cite{vovk2012conditional,lei2014distribution,barber2019limits}.

The specific procedure we suggest to construct predictive intervals with equal coverage, then, supports equitable treatment in an additional dimension. Specifically, we use the same learning algorithm for all individuals, borrowing strength from the entire population, and leveraging the entire dataset, while adjusting ``global'' predictions to make ``local'' confidence statements valid for each group. Such a training strategy may also improve the statistical efficiency of the predictive model, as illustrated by our experiments in Section \ref{sec:experiments}. Of course, our approach comes with limitations as well: we discuss these and possible extensions in Section \ref{sec:ec_limitations}.

\section{Equalized coverage} \label{sec:ec}

Let $\{(X_i, A_i, Y_i)\}$, $i = 1, \ldots, n$, be some training data where the vector $X_i \in \RR^p$ may contain the sensitive attribute $A_i \in \{0,1,2,\dots\}$ as one of the features. Imagine we hold a test point with known $X_{n+1}$ and $A_{n+1}$ and aim to construct a prediction interval $C(X_{n+1},A_{n+1}) \subseteq \RR$ which contains the unknown response $Y_{n+1} \in \RR$ with probability at least $1-\alpha$ on average within each group; here, $0 < 1-\alpha < 1$ is a desired coverage level. {(Our ideas extend to categorical responses in a fairly straightforward fashion; for the sake of simplicity, we do not consider these extensions in this paper.)} Formally, we assume that the training and test samples $ \{(X_i,A_i,Y_i)\}_{i=1}^{n+1} $ are drawn exchangeably from some arbitrary and unknown distribution $P_{XAY}$, and we wish that our  prediction interval obeys the following property:
\begin{align} \label{eq:cp_coverage}
\PP\{Y_{n+1} \in C(X_{n+1},A_{n+1}) \mid A_{n+1}=a\} \geq 1-\alpha
\end{align}
for all $a$. Above, the probability is taken over the $n$ training samples and the test case.
%\rfb{need to specify that this holds for all $a$}\yr{It appears in the next sentence, isn't it clear?}
Once more, \eqref{eq:cp_coverage} must hold for any distribution $P_{XAY}$, sample size $n$, and regardless of the group identifier $A_{n+1}$. (While this only ensures that coverage is {\em at least} $1-\alpha$ for each group---and, therefore, the groups may have unequal coverage level---we will see that under mild conditions the coverage can also be upper bounded to lie very close to the target level $1-\alpha$.)

In this section we present a methodology to achieve \eqref{eq:cp_coverage}. Our solution builds on classical conformal prediction \cite{vovk2005algorithmic,lei2018distribution} and the recent conformalized quantile regression (CQR) approach  \cite{romano2019conformalized} originally designed to construct \textbf{marginal} distribution-free prediction intervals (see also \cite{kivaranovic2019adaptive}). CQR combines the rigorous coverage guarantee of conformal prediction with the statistical efficiency of quantile regression \cite{koenker1978regression} and has been shown to be adaptive to the local variability of the data distribution under study. 
Below, we present a modification of CQR obeying~\eqref{eq:cp_coverage}. Then in Section~\ref{sec:group-conformal}, we draw connections to conformal prediction \cite{papadopoulos2002inductive,lei2018distribution} and explain how classical conformal inference can also be used to construct prediction intervals with equal coverage across protected groups.\footnote{We build on the split conformal methodology \cite{papadopoulos2002inductive,lei2018distribution} rather than its transductive (or full) version \cite{vovk2005algorithmic} due to the high computational cost of the latter. We refer the reader to \cite{vovk2005algorithmic,lei2018distribution} for more details about transductive conformal prediction, its advantages and limitations.
%\rfb{There's a terminology mixup here --- ``inductive'' = split conformal ([18] uses the term ``inductive''); ``transductive'' = full conformal}\yr{fixed}
}

% Quantile regression offers a direct approach for constructing \textbf{conditional} prediction intervals. This is done by fitting two conditional quantile functions on the training data, with upper and lower quantile levels that equal to $ \alphalo = \alpha/2$ and $ \alphahi = 1 - \alpha/2 $, say. These estimates can potentially form prediction intervals with $1-\alpha$ coverage rate. However, in practice, the estimated bands can largely under- or over-cover the test target variables. In general, there are no finite sample coverage guarantee for quantile regression and the validity of this approach is guaranteed for specific models, and under asymptotic and regularity conditions \cite{steinwart2011estimating,takeuchi2006nonparametric,meinshausen2006quantile}.

Before describing the proposed method we introduce a key result in conformal prediction, adapted to our conditional setting. Variants of the following lemma appear in the literature \cite{vovk2012conditional,vovk2005algorithmic,lei2018distribution,barber2019conformal,romano2019conformalized}.
\begin{lemma} \label{lemma:quantile-inflation}
	Suppose the random variables $Z_1,\dots,Z_{m+1}$ are exchangeable conditional on $A_{m+1}=a$, and define
	$Q_{1-\alpha}$ to be the  $(1-\alpha)(1+1/m)\text{-th empirical quantile of } \left\{Z_i : 1\leq i \leq m \right\}$.
	For any $\alpha \in (0,1)$,
	\begin{equation}
	\PP\{Z_{m+1} \leq  Q_{1-\alpha} \mid A_{m+1}=a\} \geq \alpha.
	\end{equation}
	Moreover, if the random variables $Z_1,\dots,Z_{m+1}$ are almost surely distinct, then it also holds that 
	\begin{equation}
	\PP\{Z_{m+1} \leq Q_{1-\alpha} \mid A_{m+1}=a\} \leq \alpha + 1/(m+1).
	\end{equation}
\end{lemma}
%\rfb{I think the upper bound should be $\alpha + 1/(m+1)$ (see [22])}\yr{fixed, good catch!}

\subsection{Group-conditional conformalized quantile regression (CQR)} \label{sec:group-cqr}

\begin{algorithm}[t]
	\caption{Group-conditional CQR.}	\label{alg:split_qreg}
	
	\textbf{Input:}
	\begin{algorithmic}
		\State Data $(X_i, A_i, Y_i) \in \RR^p \times \NN \times \RR, \ 1\leq i \leq n$.
		\State Nominal coverage level $1-\alpha \in (0,1)$.
		\State Quantile regression algorithm $ \mathcal{A} $.
		\State Training mode: joint/groupwise.
		\State Test point $ X_{n+1}={x} $ with sensitive attribute $A_{n+1}=a$.
	\end{algorithmic}
	
	\textbf{Process:}
	\begin{algorithmic}
		\State Randomly split $ \left\lbrace 1,\dots,n \right\rbrace  $ into two disjoint sets $ \mathcal{I}_1 $ and $ \mathcal{I}_2 $.
		\State \emph{\textbf{If}} \emph{joint training}:
		\State \quad Fit quantile functions on the whole proper training set:
		$ \left\lbrace  \hat{q}_{\alphalo}, \hat{q}_{\alphahi} \right\rbrace \leftarrow \mathcal{A}(\left\lbrace (X_i, Y_i): i \in \mathcal{I}_1 \right\rbrace) $.
		\State \emph{\textbf{Else}} use \emph{groupwise training}:
		\State \quad Fit quantile functions on the proper training examples from group $A_{n+1}=a$:
		\State \quad  $\left\lbrace  \hat{q}_{\alphalo}, \hat{q}_{\alphahi} \right\rbrace \leftarrow \mathcal{A}(\left\lbrace (X_i, Y_i): i \in \mathcal{I}_1 \text{ and } A_i=a  \right\rbrace) $.
		% 		\State Extract $\mathcal{I}_2(a)$ for group $A_{n+1}=a$, as in equation \eqref{eq:group_cal_points}.
		\State Compute $ E_i $ for each $i \in \mathcal{I}_2(a)$, as in \eqref{eq:qreg_conformity}.
		\State Compute $Q_{1-\alpha}(E, \mathcal{I}_2(a))$, the $(1-\alpha)(1+1/|\mathcal{I}_2(a)|)$-th empirical quantile of $\left\{E_i : i \in \mathcal{I}_2(a) \right\}$.
	\end{algorithmic}
	
	\textbf{Output:}
	\begin{algorithmic}
		\State Prediction interval $ \hat C(x,a) = \left[ \hat{q}_{\alphalo}(x) - Q_{1-\alpha}(E, \mathcal{I}_2(a)) , \ \hat{q}_{\alphahi}(x) + Q_{1-\alpha}(E, \mathcal{I}_2(a)) \right] $ for the unknown response $Y_{n+1}$.
	\end{algorithmic}
	
\end{algorithm}

% \rfb{Suggested change in notation for $\hat{C}(x)$ and $C(x,a)$ $\rightarrow$ $\hat{C}_{\text{init}}(x)$ and $\hat{C}(x,a)$ (or in any case, I think $C(x,a)$ should have a hat)}\yr{Fixed}

Our method starts by randomly splitting the $n$ training points into two disjoint subsets; a \emph{proper training set} $ \left\lbrace (X_i,A_i,Y_i): i \in \mathcal{I}_1 \right\rbrace  $ and a \emph{calibration set} $ \left\lbrace (X_i,A_i,Y_i): i \in \mathcal{I}_2 \right\rbrace  $. Then, consider any algorithm $\mathcal{A}$ for quantile regression that estimates conditional quantile functions from observational data, such as quantile neural networks \cite{taylor2000quantile} (described in Appendix~\ref{app:qnet}). To construct a prediction interval with $1-\alpha$ coverage, fit two conditional quantile functions on the proper training set,
\begin{align} \label{eq:learning}
\left\lbrace \hat{q}_{\alphalo}, \hat{q}_{\alphahi} \right\rbrace \leftarrow \mathcal{A}(\left\lbrace (X_i, Y_i): i \in \mathcal{I}_1 \right\rbrace),
\end{align}
at levels $ \alphalo = \alpha/2$ and $ \alphahi = 1 - \alpha/2 $, say, and form a first estimate of the prediction interval $ \hat{C}_{\text{init}}(x) = [\hat{q}_{\alphalo}(x), \ \hat{q}_{\alphahi}(x)] $ at $X=x$. $\hat{C}_{\text{init}}(x)$ is constructed with the goal that a new case with covariates $x$ should have probability $1-\alpha$ of its response lying in the interval $\hat{C}_{\text{init}}(x)$, but
the interval $\hat{C}_{\text{init}}(x)$ was empirically shown to largely under- or over-cover the test target variable \cite{romano2019conformalized}. (Quantile regression algorithms are not supported by finite sample coverage guarantees \cite{steinwart2011estimating,takeuchi2006nonparametric,meinshausen2006quantile,zhou1996direct,zhou1998statistical}.)

% These functions provide a first interval estimation $ \hat{C}(x) = [\hat{q}_{\alphalo}(x), \ \hat{q}_{\alphahi}(x)] $ at a point $X = x$. Note that while equation \eqref{eq:learning} fits a joint predictive model for all groups, it is also possible to fit a \emph{separate} regression function per each group, leading to a series of group-dependent models.

% fit low and high conditional quantile functions on the proper training examples by running a quantile regression algorithm $\mathcal{A}$ (e.g. quantile random forests \cite{meinshausen2006quantile} or quantile neural networks \cite{taylor2000quantile}):
% \begin{align} \label{eq:learning}
% 	\left\lbrace \hat{q}_{\alphalo}, \hat{q}_{\alphahi} \right\rbrace \leftarrow \mathcal{A}(\left\lbrace (X_i, Y_i): i \in \mathcal{I}_1 \right\rbrace).
% \end{align}
% Above, $\hat{q}_{\alphalo}$ and $\hat{q}_{\alphahi}$ are the estimated quantile regression models, with $ \alphalo = \alpha/2$ and $ \alphahi = 1 - \alpha/2 $, say. These functions provide a first interval estimation $ \hat{C}(x) = [\hat{q}_{\alphalo}(x), \ \hat{q}_{\alphahi}(x)]  $ at a point $X = x$. Note that while equation \eqref{eq:learning} results in a joint predictive model for all groups, it is possible to fit a \emph{separate} regression function per each group, leading to a series of group-dependent models.

This motivates the next step that borrows ideas from split conformal prediction \cite{papadopoulos2002inductive,lei2018distribution} and CQR \cite{romano2019conformalized}. Consider a group $A=a$, and compute the empirical errors (often called conformity scores) achieved by the first guess $\hat{C}_{\text{init}}(x)$. This is done by extracting the calibration points that belong to that group,
\begin{align} \label{eq:group_cal_points}
\mathcal{I}_2(a) = \{ i : i\in \mathcal{I}_2 \ \text{and} \ A_i=a\},
\end{align}
and evaluating 
\begin{align} \label{eq:qreg_conformity}
E_i := \max\{\hat{q}_{\alphalo}(X_i) - Y_i, Y_i - \hat{q}_{\alphahi}(X_i)\}, \quad i \in \mathcal{I}_2(a).
\end{align}
This step provides a family of conformity scores $ \left\lbrace E_i: i \in \mathcal{I}_2(a) \right\rbrace  $ that are restricted to the group $A=a$. Each score 
measures the signed distance of the target variable $Y_i$ to the boundary of the interval $\hat{C}_{\text{init}}(x)$; if $Y_i$ is located outside the initial interval, then $E_i>0$ is equal to the distance to the closest interval endpoint. If $Y_i$ lies inside the interval, then $E_i\leq0$ and its magnitude also equals the distance to the closest endpoint. As we shall see immediately below, these scores may serve to measure the quality of the initial guess $\hat{C}_{\text{init}}(\cdot)$ and used to calibrate it as to obtain the desired distribution-free coverage.
%Crucially, our approach makes no assumptions on the model class of $\hat{C}_{\text{init}}(\cdot)$, its accuracy or coverage.
Crucially, our approach makes no assumptions on the form or the properties of $\hat{C}_{\text{init}}(\cdot)$---it may come from any model class, and is not required to meet any particular level of accuracy or coverage. Its role is to provide a ``base algorithm'' around which we will build our predictive intervals.

%\rfb{I think we should state explicitly that we make no assumptions on the accuracy / coverage / model class / etc of $\hat{C}_{init}$}\yr{added}

% In words, suppose that $A$ is a $\{0,1\}$ variable, the last step results in two distinct sets of conformity scores $ \left\lbrace E_i: i \in \mathcal{I}_2^0 \right\rbrace  $ and $ \left\lbrace E_i: i \in \mathcal{I}_2^1 \right\rbrace  $, where each element measures the signed distance of the target variable $Y_i$ to its closest conditional quantile estimate.

Finally, the following crucial step builds a prediction interval for the unknown $ Y_{n+1} $ given $X_{n+1}=x$ and $A_{n+1}=a$. This is done by computing
\begin{align} \label{eq:qreg_quantilde_split}
Q_{1-\alpha}(E, \mathcal{I}_2(a)) :=  (1-\alpha)(1+1/|\mathcal{I}_2(a)|)\text{-th empirical quantile of} \left\{E_i : i \in \mathcal{I}_2(a) \right\},
\end{align}
which is then used to calibrate the first interval estimate as follows:
\begin{align} \label{eq:c_cplit_qreg}
\hat C(x, a) := \left[ \hat{q}_{\alphalo}(x) - Q_{1-\alpha}(E, \mathcal{I}_2(a)) , \ \hat{q}_{\alphahi}(x) + Q_{1-\alpha}(E, \mathcal{I}_2(a)) \right].
\end{align}

Before proving the validity of the interval in \eqref{eq:c_cplit_qreg}, we pause to present two possible training strategies for the initial quantile regression interval $\hat{C}_{\text{init}}(x)$. We refer to the first as \emph{joint training} as it uses the whole proper training set to learn a predictive model, see \eqref{eq:learning}. 
The second approach, which we call \emph{groupwise training}, constructs a prediction interval for $Y_{n+1}$ separately for each group; that is, for each value $a=0,1,2,\dots$, we fit a regression model to all training examples with $A_{n+1}=a$.
%The second approach we call \emph{groupwise training} constructs a prediction interval for $Y_{n+1}$ by fitting a regression model to proper training examples that share the same group identifier $A_{n+1}=a$.
These two variants of the CQR procedure are summarized in Algorithm~\ref{alg:split_qreg}.  While the statistical efficiency of the two approaches can differ (as we will see in Section \ref{sec:experiments}), both are guaranteed to attain valid group-conditional coverage for any data distribution and regardless of the choice or accuracy of the quantile regression estimate.

% The next crucial step calibrates the plug-in interval per each group. This is done by computing
% \begin{align} \label{eq:qreg_quantilde_split}
% Q_{1-\alpha}(E, \mathcal{I}_2^a) :=  (1-\alpha)(1+1/|\mathcal{I}_2^a|)\text{-th empirical quantile of} \left\{E_i : i \in \mathcal{I}_2^a \right\},
% \end{align}
% which is then used to construct
% \begin{align} \label{eq:c_cplit_qreg}
% C(X_{n+1}, A_{n+1}) = \left[ \hat{q}_{\alphalo}(X_{n+1}) - Q_{1-\alpha}(E, \mathcal{I}_2^a) , \ \hat{q}_{\alphahi}(X_{n+1}) + Q_{1-\alpha}(E, \mathcal{I}_2^a) \right],
% \end{align}
% where $X_{n+1}$ is an observed test point with a known sensitive attribute $A_{n+1}=a$. As stated in Theorem \ref{thm:validity_qreg} below, the interval in \eqref{eq:c_cplit_qreg} attains valid group-conditional coverage. Crucially, the following statement holds in finite sample, for any distribution and regardless of the choice or accuracy of the quantile regression estimate.

\begin{theorem} \label{thm:validity_qreg}
	If $ (X_i, A_i, Y_i) $, $ i=1,\dots, n+1 $ are exchangeable, then the prediction interval $ \hat{C}(X_{n+1},A_{n+1}) $ constructed by Algorithm \ref{alg:split_qreg} obeys 
	\begin{equation}
	\PP\{Y_{n+1} \in \hat C(X_{n+1}, A_{n+1}) \mid A_{n+1}=a \} \geq 1-\alpha
	\end{equation}
	{for each group $a=0,1,2,\dots$.}
	Moreover, if the conformity scores $\{E_i : \ i\in\mathcal{I}_2(a) \cup \{n+1\}\}$ for $A_{n+1}=a$ are almost surely distinct, then the group-conditional prediction interval is nearly perfectly calibrated:
	\begin{equation}
	\PP\{Y_{n+1} \in \hat C(X_{n+1}, A_{n+1}) \mid A_{n+1}=a\} \leq 1-\alpha+\frac{1}{|\mathcal{I}_2(a)|+1}
	\end{equation}
	{for each group $a=0,1,2,\dots$.}
\end{theorem}
\begin{proof}
	%Consider a test triplet $ (X_{n+1}, A_{n+1}=a, Y_{n+1}) $. 
	{Fix any group $a$.}
	Since our calibration samples are exchangeable, the conformity scores \eqref{eq:qreg_conformity} $ \{E_{i} : \ i \in \mathcal{I}_2(a)\}$ are also exchangeable. Exchangeability also holds when we add the test score $E_{n+1}$ to this list. Consequently, by Lemma~\ref{lemma:quantile-inflation},
	\begin{equation}
	\label{eq:En+1}
	1-\alpha \leq \PP(E_{n+1} \leq Q_{1-\alpha}(E, \mathcal{I}_2(a)) \mid A_{n+1}=a) \leq 1-\alpha+\frac{1}{|\mathcal{I}_2(a)|+1},
	\end{equation}
	where the upper bound holds under the additional assumption that $ \{E_{i} : \ i \in \mathcal{I}_2(a) \cup \{n+1\}\}$ are almost surely distinct, while the lower bound holds without this assumption. 
	
	To prove the validity of $ \hat C(X_{n+1},A_{n+1}) $ conditional on $A_{n+1}=a$, observe that, by definition,
	\[
	Y_{n+1} \in \hat C(X_{n+1},A_{n+1}) \quad \text{if and only if} \quad E_{n+1} \le Q_{1-\alpha}(E, \mathcal{I}_2(a)).
	\]
	Hence, the result follows from \eqref{eq:En+1}.

\end{proof}

\paragraph{Variant: asymmetric group-conditional CQR}

When the distribution of the conformity scores is highly skewed, the coverage error may spread asymmetrically over the left and right tails. In some applications it may be better to consider a variant of \mbox{Algorithm \ref{alg:split_qreg}} that controls the coverage of the two tails separately, leading to a stronger conditional coverage guarantee. To achieve this goal, we follow the approach from \cite{romano2019conformalized} and evaluate two separate empirical quantile functions:  one for the left tail,
\begin{align} \label{eq:qreg_quantilde_split_low_tail}
Q_{1-\alphalo}(E_{\mathrm{lo}}, & \mathcal{I}_2(a)) := \\ & (1-\alphalo)(1+1/|\mathcal{I}_2(a)|)\text{-th empirical quantile of} \left\{\hat{q}_{\alphalo}(X_i) - Y_i : i \in \mathcal{I}_2(a) \right\};
\end{align}
and the second for the right tail
\begin{align} \label{eq:qreg_quantilde_split_high_tail}
Q_{1-\alphahi}(E_{\mathrm{hi}}, & \mathcal{I}_2(a)) :=  \\ & (1-\alphahi)(1+1/|\mathcal{I}_2(a)|)\text{-th empirical quantile of} \left\{Y_i - \hat{q}_{\alphahi}(X_i): i \in \mathcal{I}_2(a) \right\}.
\end{align}
Next, we set $\alpha = \alphalo + \alphahi$ and construct the interval for $ Y_{n+1} $ given $X_{n+1}=x$ and $A_{n+1}=a$:
\begin{equation} \label{eq:two_tails_interal}
\hat C(x,a) := [\hat{q}_{\alphalo}(x) - Q_{1-\alphalo}(E_{\mathrm{lo}}, \mathcal{I}_2(a)),\; \hat{q}_{\alphahi}(x) + Q_{1-\alphahi}(E_{\mathrm{hi}}, \mathcal{I}_2(a))].
\end{equation}
The validity of this procedure is stated below.

\begin{theorem} \label{thm:validity_qreg_asymmetric}
	Suppose the samples $ (X_i, A_i, Y_i) $, $ i=1,\dots, n+1 $ are exchangeable. With the notation above, put $\operatorname{lower}(X_{n+1}) = \hat{q}_{\alphalo}(X_{n+1}) - Q_{1-\alphalo}(E_{\mathrm{lo}},\mathcal{I}_2(A_{n+1}))$ and $\operatorname{upper}(X_{n+1}) = 
	\hat{q}_{\alphahi}(X_{n+1}) + Q_{1-\alphahi}(E_{\mathrm{hi}}, \mathcal{I}_2(A_{n+1}))$ for short. Then 
	\[
	%\label{eq:asymmetric-coverage-lower}
	1-\alphalo \leq \PP\{Y_{n+1} \geq \operatorname{lower}(X_{n+1})  \mid A_{n+1}=a\} 
	\leq{1-\alphalo+\frac{1}{|\mathcal{I}_2(a)|+1}}
	\]
	and
	\[
	% \label{eq:asymmetric-coverage-upper}
	1-\alphahi \leq \PP\{Y_{n+1} \leq \operatorname{upper}(X_{n+1})  \mid A_{n+1}=a\}\\ \leq{1-\alphahi+\frac{1}{|\mathcal{I}_2(a)|+1}},
	\]
	where the lower bounds above always hold while the upper bounds hold under the additional assumption that the residuals are almost surely distinct. 
	%The lower bounds above always hold whereas the upper bounds hold under the additional assumption that the residuals are almost surely distinct.
	Under these circumstances, the interval  \eqref{eq:two_tails_interal} obeys
	\begin{equation} \label{eq:asymmetric-interval}
	1 - (\alphalo + \alphahi) \leq \PP\{Y_{n+1} \in \hat C(X_{n+1}, A_{n+1}) \mid A_{n+1}=a \} \leq{1-(\alphalo + \alphahi)+\frac{2}{|\mathcal{I}_2(a)|+1}}. 
	\end{equation}
\end{theorem}
\begin{proof}
	As in the proof of Theorem~\ref{thm:validity_qreg}, the validity of the lower and upper bounds is obtained by applying Lemma~\ref{lemma:quantile-inflation} twice.
\end{proof}

\subsection{Group-conditional conformal prediction} \label{sec:group-conformal}

The difference between CQR \cite{romano2019conformalized} and split conformal prediction \cite{papadopoulos2002inductive} is that the former calibrates an estimated quantile regression interval $\hat{C}_{\text{init}}(X)$, while the latter builds a prediction interval around an estimate of the conditional mean $\hat Y = \hat \mu(X)$. For instance, $\hat \mu$ can be formulated as a \emph{classical regression} function estimate, obtained by minimizing the mean-squared-error loss over the proper training examples. To construct predictive intervals for the group $A=a$, then simply replace both $\hat{q}_{\alphalo}$ and $\hat{q}_{\alphahi}$ with $\hat{\mu}$ in \mbox{Algorithm \ref{alg:split_qreg}} (or in its two-tailed variant). The theorems go through, and this procedure gives predictive intervals with exactly the same guarantees as before. As we will see in our empirical results, a benefit of explicitly modeling quantiles is superior statistical efficiency. 

\section{Case study: predicting utilization of medical services} \label{sec:experiments}

The Medical Expenditure Panel Survey (MEPS) 2016 data set \cite{meps_21}, provided by the Agency for Healthcare Research and Quality, contains information on individuals and their utilization of medical services. The features used for modeling include age, marital status, race, poverty status, functional limitations, health status, health insurance type, and more. We split these features into dummy variables to encode each category separately. The goal is to predict the health care system utilization of each individual; a score that reflects the number of visits to a doctor's office, hospital visits, etc. After removing observations with missing entries, there are $n=15656$ observations on $p=139$ features. We set the sensitive attribute $A$ to \textbf{race}, with $A=0$ for non-white and $A=1$ for white individuals, resulting in $n_0=9640$ samples for the first group and $n_1=6016$ for the second. In all experiments we transform the response variable by $Y = \log(1 + (\text{utilization score}))$ as the raw score is highly skewed.

Below, we illustrate that empirical quantiles can be used to detect prediction bias. Next, we show that usual (marginal) conformal methods do not attain equal coverage across the two groups. Finally, we compare the performance of joint vs.~groupwise model fitting and show that, in this example, the former yields shorter predictive intervals. 
% The joint modeling approach follows equation~\eqref{eq:learning} and fits a single predictive model for both groups using the whole proper training set. In contrast, the separate modeling approach constructs prediction intervals for group $A=a$ by (i) fitting a regression model to the proper training examples that belong to this specific group, and (ii) conformalzing the regression estimate to form valid intervals.

% We also compare between joint and separate model fitting and show that the former yields shorter intervals. 

\begin{figure}[t]
	\centering
	\includegraphics[width=0.50\textwidth]{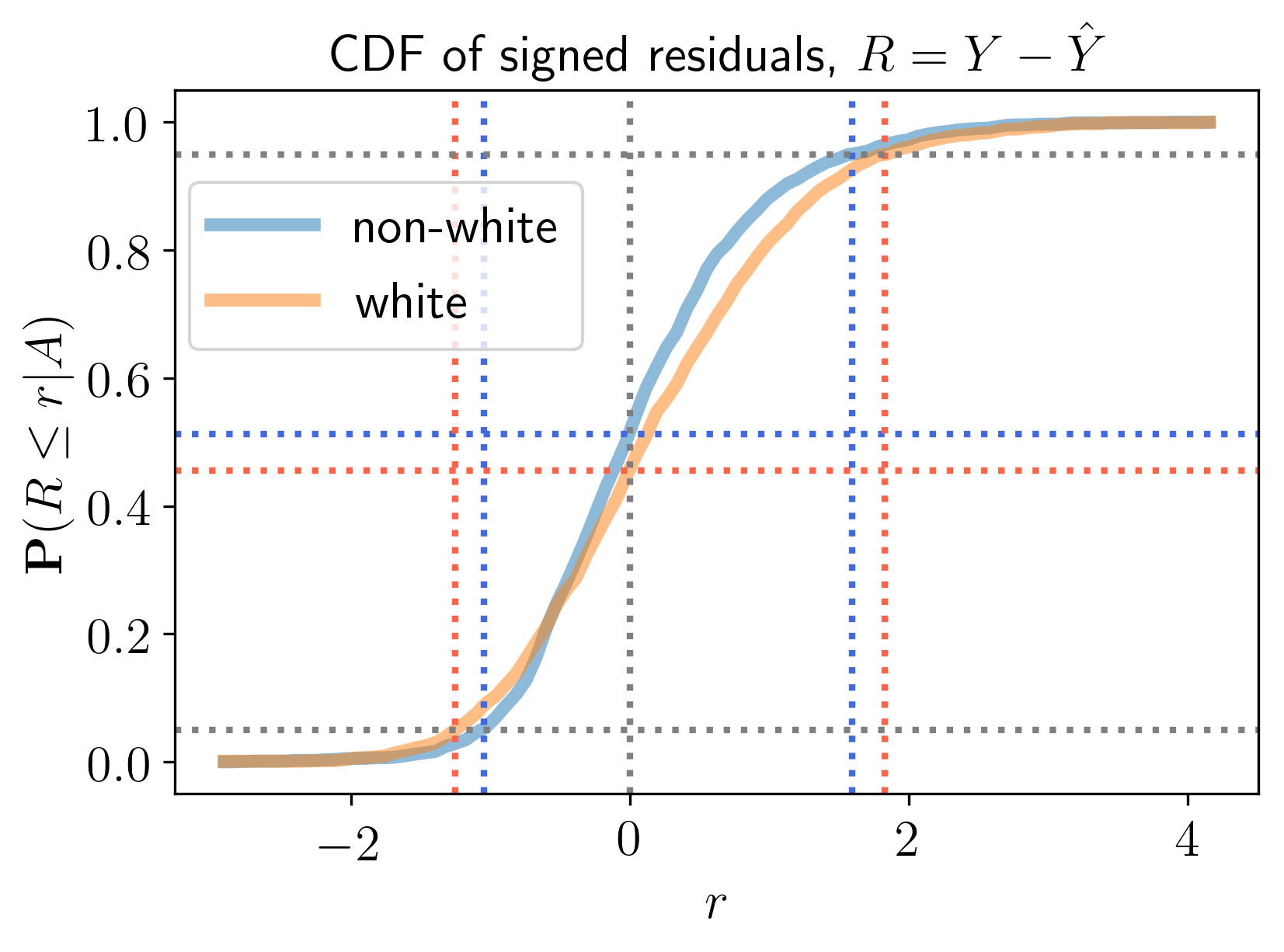}
	\caption{Empirical cumulative distribution function of the signed residuals $R = Y - \hat{Y}$ for each sensitive attribute, 
		computed on test samples. 
		%  	light blue -- non-white individuals, orange -- white individuals. 
		The blue dashed horizontal line is the value  of $\PP\{Y \leq \hat{Y} | A=0\}$ equal to $0.51$. Similarly, the red dashed horizontal line is $\PP\{Y \leq \hat{Y} | A=1\} = 0.45$.  The dashed vertical colored lines present the 0.05th and 0.95th quantiles of each group, defined in \eqref{eq:emp_low_q} and \eqref{eq:emp_high_q}, respectively. Left side: $r^{\text{lo}}_0=-1.04$ (blue), $r^{\text{lo}}_1=-1.25$ (red). Right side: $r^{\text{hi}}_0=1.59$ (blue) and $r^{\text{hi}}_1=1.83$ (red).} 
	\label{fig:meps_median}
\end{figure}

% \begin{figure}
% 	\centering
% 	\includegraphics[width=0.60\textwidth]{./figures/CDF_lo_hi_q.png}
% 	\caption{Empirical low and high quantiles of test residuals (see caption of Figure \ref{fig:meps_median} for more details). The left blue dashed vertical line corresponds to $r^{\text{lo}}_0=-1.06$ for which $\PP\{R \leq r^{\text{lo}}_0 | A=0\} \geq 0.05$. Similarly, the left red dashed line presents the probability $\PP\{R \leq r^{\text{lo}}_1 | A=1\} \geq 0.05$, with $r^{\text{lo}}_1=-1.29$. Analogously, the colored vertical dashed lines on the right reflect the values of $r$ that correspond to the high 0.95th quantiles; blue: $r^{\text{hi}}_0=1.61$; red: $r^{\text{hi}}_1=1.83$.
% 	} 
% 	\label{fig:meps_low_high_q}
% \end{figure}

\subsection{Bias detection} \label{sec:bias_detection}

We randomly split the data into train (80\%) and test (20\%) sets and standardized the features to have zero mean and unit variance; the means and variances are computed using the training examples. Then we fit a neural network regression function $\hat{\mu}$ on the training set, where the network architecture, optimization, and hyper-parameters are similar to those described and implemented in \cite{romano2019conformalized}. The code for reproducing all the experiments is available online at \url{https://github.com/yromano/cqr}. Next, we compute the signed residuals of the test samples,
\begin{align}
R_i = Y_i - \hat{Y}_i,
\end{align}
where $\hat{Y}_i = \hat{\mu}(X_i)$, and plot the resulting empirical cumulative distribution functions $\PP\{R \leq r | A=0\}$ and $\PP\{R \leq r | A=1\}$ in Figure \ref{fig:meps_median}. Observe that $\PP\{R \leq r | A=0\} \neq \PP\{R \leq r | A=1\}$. In particular, when comparing the two functions at $r=0$, we see that $\hat{\mu}$ overestimates the response of the non-white group and underestimates the response of the white group, as
\begin{align}
\PP\{Y \leq \hat{Y} | A=0\} = 0.51 > 0.45 = \PP\{Y \leq \hat{Y} | A=1\}.
\end{align}

Recall that the lower and upper quantiles of the signed residuals are used to construct valid group-conditional prediction intervals. While these must be evaluated on calibration examples (see next section), for illustrative purposes we present below the 0.05th and 0.95th quantiles of each group using the two cumulative distribution functions of test residuals. To this end, we denote by $ r^{\text{lo}}_0 $ and $ r^{\text{lo}}_1 $ the lower empirical quantiles of the non-white and white groups, defined to be the smallest numbers obeying the relationship
%We also compute the low empirical quantiles of each group, denoted by $ r^{\text{lo}}_0 $ and $ r^{\text{lo}}_1 $, and defined to be the smallest numbers obeying the relationship 
\begin{equation} \label{eq:emp_low_q}
\PP\{R \leq r^{\text{lo}}_0 | A = 0\} \geq 0.05 \quad \text{and} \quad \PP\{R \leq r^{\text{lo}}_1 | A = 1\} \geq 0.05.
\end{equation}
Following Figure \ref{fig:meps_median}, this pair is equal to
\begin{equation}
r^{\text{lo}}_0 = -1.04 > -1.25 = r^{\text{lo}}_1,
\end{equation}
implying that for at least 5\% of the test samples of each group, the fitted regression function $\hat{\mu}$ overestimates the utilization of medical services with larger errors for white individuals than for non-white individuals. %subjects.

As for the upper empirical quantiles,
%\rfb{I prefer the terms ``lower'' and ``upper'' rather than ``low'' and ``high'', throughout (I think the notation $\alpha_{lo}$ and $\alpha_{hi}$ is fine in either case though)}\yr{fixed (here and everywhere)}
we compute the smallest $ r^{\text{hi}}_0 $ and $ r^{\text{hi}}_1 $ obeying
\begin{equation} \label{eq:emp_high_q}
\PP\{R \leq r^{\text{hi}}_0 | A = 0\} \geq 0.95 \quad \text{and} \quad \PP\{R \leq r^{\text{hi}}_1 | A = 1\} \geq 0.95,
\end{equation}
and obtain 
\begin{equation}
r^{\text{hi}}_0 = 1.59 < 1.83 = r^{\text{hi}}_1.
\end{equation}
Here, in order to cover the target variable for white individuals at least 95\% of the time we should inflate the regression estimate by an additive factor equal to $1.83$. For non-white individuals, the additive factor is \emph{smaller} and equal to $1.59$. This shows that $\hat{\mu}$ systematically predicts higher utilization of non-white individuals relative to white individuals.

%This shows, once again, that $\hat{\mu}$ is biased, predicting higher utilization of non-white individuals relative to white subjects.

\subsection{Achieving equalized coverage} \label{sec:ec_experiments}

We now verify that our proposal constructs intervals with equal coverage across groups. Below, we set $\alpha = 0.1$. To avoid the coverage errors to be spread arbitrarily over the left and right tails, we choose to control the two tails independently by setting $\alphalo = \alphahi = \alpha/2 = 0.05$ in \eqref{eq:two_tails_interal}.
%, where the target miscoverage rate is fixed in all experiments and equal to $\alpha=0.1$.
We arbitrarily set the size of the proper training and calibration sets to be identical. (The features are standardized as discussed earlier.) 

%\rfb{This paragraph is new, replacing a shorter version---I felt we needed a bit more detail}
For our experiments, we test six different methods for producing conformal predictive intervals. 
We compare two types of constructions for the predictive interval:
\begin{itemize}
    \item Conformal prediction (CP),
    %\rfb{For consistency maybe the table and text should write CP vs CQR, not ``conformal'' or ``conformal prediction'' vs CQR}\yr{now we use CP and CQR in the table}
    where the predictive interval is built around an estimated mean $\hat{\mu}$ (as described in Section~\ref{sec:group-conformal});
    \item Conformal quantile regression (CQR), where the predictive interval is constructed around initial estimates $\hat{q}_{\alpha_{\text{lo}}}$ and $\hat{q}_{\alpha_{\text{hi}}}$ of the lower and upper quantiles.
\end{itemize}
In both cases, we use a neural network to construct the models; we train the models using the software provided by \cite{romano2019conformalized}, using the same neural network design and learning strategy. For both the CP and CQR constructions, we then implement three versions:
\begin{itemize}
    \item Marginal coverage, where the intervals $\hat{C}(X)$ are constructed by pooling all the data together rather than splitting into subgroups according to the value of $A$;
    \item Conditional coverage with groupwise models, where the initial model for the mean $\hat{\mu}$ or for the quantiles $\hat{q}_{\alpha_{\text{lo}}},\hat{q}_{\alpha_{\text{hi}}}$ is constructed separately for each group $A=0$ and $A=1$;
    \item Conditional coverage with a joint model, where the initial model for the mean $\hat{\mu}$ or for the quantiles $\hat{q}_{\alpha_{\text{lo}}},\hat{q}_{\alpha_{\text{hi}}}$ is constructed pooling data across both groups $A=0$ and $A=1$.
\end{itemize}

The results are summarized in Table \ref{tab:results}, displaying the average length and coverage of the marginal and group-conditional conformal methods. These are evaluated on \emph{unseen test data} and averaged over 40 train-test splits, where 80\% of the samples are used for training (the calibration examples are a subset of the training data) and 20\% for testing. 
All the conditional methods perfectly achieve 90\% coverage per group (this is a theorem after all). On the other hand, the marginal CP method under-covers in the white group and over-covers in the non-white group (interestingly, though, the marginal CQR method attains equalized coverage even though it is not designed to give such a guarantee).
%\rfb{I think we should remove this sentence: ``This highlights the advantage of quantile regression, which can potentially construct an approximately valid prediction interval for a given value of the covariates $X=x$.''}
% Notice how both \textbf{Marginal Neural Network} and \textbf{Marginal CQR Neural Network} undercover the response variables of the white group. These two marginal methods also overcover the target variables of the non-white group. 

Turning to the statistical efficiency of the conditional conformal methods, we see that conditional CQR outperforms conditional CP in that it constructs shorter and, hence, more informative intervals, especially for the non-white group. The table also shows that the intervals for the white group are wider than those for the non-white group across all four conditional methods, and that joint model fitting is here more effective than groupwise model fitting as the former achieves shorter prediction intervals.
% This conclusion aligns with the low and high empirical quantiles displayed in Figure \ref{fig:meps_low_high_q} (although measured on test data). 

\begin{table}[t]
	\centering
	\begin{tabular}{@{}llcc@{}}
		\toprule
		\multicolumn{1}{c}{Method}                             & \multicolumn{1}{c}{Group} & \multicolumn{1}{c}{Avg. Coverage} & Avg. Length \\ \midrule
		\multirow{2}{*}{*Marginal CP}                  & Non-white                 & 0.920                             & 2.907       \\
		& White                     & 0.871                             & 2.907       \\
		\cmidrule(l){1-4}
		\multirow{2}{*}{Conditional CP (groupwise)}     & Non-white                 & 0.903                             & 2.764       \\
		& White                     & 0.901                             & 3.182       \\
		\cmidrule(l){1-4}
		\multirow{2}{*}{Conditional CP (joint)}        & Non-white                 & 0.904                             & 2.738       \\
		& White                     & 0.902                             & 3.150       \\
		\cmidrule(l){1-4}
		\multirow{2}{*}{*Marginal CQR}              & Non-white                 & 0.905                             & 2.530       \\
		& White                     & 0.894                             & 3.081       \\
		\cmidrule(l){1-4}
		\multirow{2}{*}{Conditional CQR (groupwise)} & Non-white                 & 0.904                             & 2.567       \\
		& White                     & 0.900                             & 3.203       \\
		\cmidrule(l){1-4}
		\multirow{2}{*}{Conditional CQR (joint)}    & Non-white                 & 0.902                             & {2.527}       \\
		& White                     & 0.901                             & {3.102}       \\ 
		\bottomrule
	\end{tabular}
	
	\vspace{10pt}
	\caption{Length and coverage of both marginal and group-conditional prediction intervals ($ \alpha =0.1$) constructed by conformal prediction (CP) and CQR for MEPS dataset \cite{meps_21}. The results are averaged across 40 random train-test (80\%/20\%) splits. \emph{Groupwise} -- two independent predictive models are used, one for non-white and another for white individuals; \emph{joint} -- the same predictive model is used for all individuals. In all cases, the model is formulated as a neural network. The methods marked by an asterisk are not supported by a group-conditional coverage guarantee. %\rfb{The method names in this caption and in the table are no longer in agreement with the names in the text}\yr{fixed}
	}
	\label{tab:results}
\end{table}

\section{Discussion} \label{sec:ec_limitations}

% {\color{blue}
% EC does not supplant equalized odds; here we would highlight a second illustrative example where we would still want equalized odds because equalized coverage would not cut it.

% Options:
% \begin{itemize}
%     \item Different sample size per group may result in different interval width, reflecting the limitation of the predictive model. Rina's second example could fit nicely here.
%     \item Discussion on Rina's first example. Equalized odds will put utility on the table. We are not decision makers.
% \end{itemize}
% }

\subsection{Larger intervals for a subpopulation}

It is possible that the intervals constructed with our procedure have different lengths across groups.
%by biased predictive models will tend to have different average length across groups. 
For example, our experiments show that, on the average, the white group has wider intervals than the non-white group.
One might argue that the different length distribution is in itself a type of unfairness.
%One might argue that the different length distribution is in itself a cause of unfairness and that a fair algorithm should avoid it
%\rfb{I think this sentence might be confusing ... aren't we trying to say something more along the lines of, in a perfect world we might be able to provide same length intervals for everyone, but given the available data it might be unavoidable? I would shorten the sentence to just say ``One might argue that the different length distribution is in itself a cause of unfairness.'' (or maybe a ``type'' of unfairness.} \yr{fixed}
We believe that when there is a difference in performance across the protected groups, one needs to understand why this is so. In some cases this difference might be reduced by improving the predictive algorithm, collecting more data for the minority population, introducing new features with higher predictive power, and so on. It may also be the case that higher predictive precision in one group versus another may arise from bias, whether intentional or not, in the type of model we use, the choice of features we measure, or other aspects of our regression process---e.g., if historically more emphasis was placed on finding accurate models for a particular group $a$, then we may be measuring features that are highly informative for prediction within group $a$ while another group $a'$ would be better served by measuring a different set of variables. Crucially, %we are not proposing an intervention or a new policy to alleviate this issue, as our focus is in statistical inference and what we need to guarantee there and not in automatic decision making. If the decision maker is disappointed to see an interval that is long, this means that the predictive model cannot say very much.
we do not want to mask this differential in information, but rather make it explicit---thereby possibly motivating the decision makers to take action.

% Pressing this point further, let's consider the following decision, based on perfect knowledge (no need of statistical inference). We might want to admit to college a student from a minority group with even if we know his GPA will be among the lowest, and he will need extra resources, simply because his disadvantageous background and we believe that by giving him (and others like him) the opportunity to go to college, we will eventually reduce the racial disparity and inequality that plague our society. To automate decisions of this type, we would need to specify a set of preferences, multi-generational outcomes, and more, that go well beyond predicting the GPA of a given student. Note that people, depending on their preferences and their goals, would disagree on what is the appropriate decision here, and that is why it is presumptuous to automate decisions for others.

% \yradd{
% \subsection{Same predictive model for all individuals} \label{sec:joint_vs_groupwise}

% Another dimension in which equalized coverage supports equitable treatment is that we have an opportunity to use the same learning algorithm for all individuals, borrowing strength from the entire population, while adjusting ``global'' predictions to make ``local'' predictions valid for each group. As illustrated in our experiments, such a training strategy may even improve the statistical efficiency of the predictive model.
% }

\subsection{The use of protected attribute}

% \rfb{I am leaning towards asking whether we are certain we want to include this section. This is obviously an important issue but it's quite far outside the scope of our study and I'm not sure we can do it justice here.} \ejc{I believe this section is *extremely* important.}

The debate around fairness in general, and our proposal in particular, requires the definition of ``classes'' of individuals across which we would like an unbiased treatment. In some cases these coincide with ``protected attributes'' where discrimination on their basis is prohibited by the law. The  legislation sometimes does not allow the decision maker to know/use the protected attribute in reaching a conclusion, as a measure to caution against discrimination. While ``no discrimination'' is a goal everyone should embrace regardless whether the law mandates it or not, the opportunity of using ``protected attributes'' in data-driven recommendation systems is another matter. On the one hand, not relying on protected attributes is certainly not sufficient to guarantee absence of discrimination (see, e.g., \cite{dwork2012fairness,hardt2016equality,dieterich2016compas,corbett2018measure,buolamwini2018gender,gardner2019evaluating,zafar2017fairness,chouldechova2017fair}).
On the other hand, information on protected attributes might be necessary to guarantee equitable treatment.
Our procedure relies on the knowledge of protected attributes, so we want to expand on this last point a little.
In absence of knowledge of what are the causal determinants of an outcome, ``protected attributes'' can be an important component of a predictor. To quote from  \cite{corbett2018measure}: ``in the criminal justice system, for example, women are typically less likely to commit a future violent crime than men with similar criminal histories. As a result, gender-neutral risk scores can systematically overestimate a woman's recidivism risk, and can in turn encourage unnecessarily harsh judicial decisions. Recognizing this problem, some jurisdictions, like Wisconsin, have turned to gender-specific risk assessment tools to ensure that estimates are not biased against women.''
For disease risk assessment (Example~\ref{item:disease_risk} earlier) or 
related tasks such as diagnosis and drug prescription, race often provides relevant information and is routinely used. Presumably, once we understand the molecular basis of diseases and drug responses, and once sufficiently accurate measurements on patients are available, race may cease to be useful. Given present circumstances, however, Risch et al. \cite{RetT02} argue that ``identical treatment is not equal treatment'' and that ``a race-neutral or color-blind approach to biomedical research is neither equitable nor advantageous, and would not lead to a reduction of disparities in disease risk or treatment efficacies between groups.'' In our context, the use of protected attributes allows a rigorous evaluation of the potentially biased performance for different groups. Clearly, our current proposal can be adopted only when data on protected attributes has been collected and generalizations are topic for further research.

\subsection{Conclusion and future work} \label{sec:conclusion}

% {\color{blue}
% \begin{itemize}
%     \item The need for uncertainty estimation. The need for equalized coverage.
%     \item Combining equalized odds (or any other fairness measure) with equalized coverage. Fit your black box with fairness constraints.
%     \item What if we cannot use the sensitive attribute?
%     \item What if the sensitive attribute is not discrete?
% \end{itemize}
% }

%We propose moving away from traditional fairness criteria by 
We add to the tools that support fairness in data-driven recommendation systems by 
developing a highly operational
%meta-algorithm
method
that can augment any prediction rule
with the best available \emph{unbiased} uncertainty estimates across groups. This is achieved by constructing prediction intervals that attain valid coverage regardless of the value of the sensitive attribute. The method is supported by rigorous coverage guarantees, as demonstrated on real data examples. Although the focus of this paper is on continuous response variables, one can adapt tools from conformal inference \cite{vovk2005algorithmic} to construct prediction sets with equalized coverage for categorical target variables as well.

In this paper, we have not discussed other measures of fairness. In a future publication, we intend to explore the relationship between equalized coverage
%\rfb{Should we be consistent and always use ``equalized coverage''?}\yr{fixed here and in other places that the use of ``equalized'' instead of ``equal'' makes sense for me}
and existing fairness criteria.

% \yr{I'm not sure if the following should be mentioned} Beyond fairness, our operationalization of equal treatment across groups has applications to many of our data domains. Indeed, each domain draws data from multiple heterogeneous sources---satellite and other imagery, panel surveys and web scraping, social and traditional media---so that building prediction intervals that are accurate \emph{across} modalities will be both important and presents an attractive long-term goal with applications far beyond fairness.

\subsection*{Acknowledgements}
E.~C. was partially supported by the Office of Naval Research (ONR)
under grant N00014-16- 1-2712, by the Army Research Office (ARO) under
grant W911NF-17-1-0304, by the Math + X award from the Simons
Foundation and by a generous gift from TwoSigma. Y.~R.~was partially supported by the
ARO grant and by the same Math + X award.  Y.~R.~thanks the Zuckerman
Institute, ISEF Foundation and the Viterbi Fellowship, Technion, for providing additional research support. R.~F.~B.~was partially supported National Science Foundation via grant DMS-1654076 and by an Alfred P. Sloan fellowship. C.~S.~was partially supported National Science Foundation via grant DMS-1712800.

\bibliographystyle{ieeetr}
\bibliography{MyBib}

%\bibliographystyle{unsrt}
%\bibliography{MyBib}

\appendix

\section{Quantile Neural Networks} \label{app:qnet}

% We follow \cite{koenker1978regression} and formulate this estimation problem as a learning task:
% \begin{align} \label{eq:q_optimization}
% \hat{q}_\alpha(x) = f(x; \hat\theta), \qquad \hat\theta = \underset{\theta}{\mathrm{argmin}} \ \frac{1}{|\mathcal{I}_1|} \sum_{i \in \mathcal{I}_1} \rho_\alpha (Y_i - f(X_i ; \theta))  + \mathcal{R}(\theta),
% \end{align}
% where $ \rho_\alpha $ is the pinball loss \cite{koenker1978regression,steinwart2011estimating}, defined by
% \begin{align} \label{eq:pinball}
% \rho_\alpha(z) := \begin{cases}
% \alpha z & \text{if } z> 0, \\
% (\alpha-1)z & \text{otherwise};
% \end{cases}
% \end{align} 
% $f(x;\theta)$ is a regression function that, in our experiments, is formulated as a deep neural network with parameters $\theta$, and $\mathcal{R}(\theta)$ is a possible regularizer. Suppose that the sensitive attribute $A$ is contained in $X$. The functions $\hat{q}_{\alphalo}$ and $ \hat{q}_{\alphahi}$ can potentially form to a prediction interval with conditional $1-\alpha$ coverage rate, however, in practice $ \hat{C}(x) $ was shown to largely under- or over-cover the test target variables \cite{romano2019conformalized}. Furthermore, in general, this approach is not supported by finite sample coverage guarantee \cite{steinwart2011estimating,takeuchi2006nonparametric,meinshausen2006quantile}.

We follow \cite{koenker1978regression} and cast the estimation problem of the conditional quantiles of $ Y | X {=} x $ as an optimization problem. Given training examples $\{(X_i,Y_i) : i \in \mathcal{I}_1\}$, we fit a parametric model using the pinball loss \cite{koenker1978regression,steinwart2011estimating}, defined by
\begin{align}
\rho_\alpha(y - \hat y) := \begin{cases}
\alpha (y - \hat y), & \text{if } y - \hat y > 0, \\
(\alpha-1) (y - \hat y) ,& \text{otherwise},
\end{cases}
\end{align} 
where $\hat y$ is the output of a regression function $\hat q_\alpha(x)$ formulated as a deep neural network. The network design and training algorithm are
% We follow \cite{koenker1978regression} and cast the estimation problem of the conditional quantiles of $ Y | X {=} x $ as an optimization problem. Given training examples $\{(X_i,Y_i) : i \in \mathcal{I}_1\}$, we fit a parametric model by
% \begin{align} \label{eq:q_optimization}
% \hat{q}_\alpha(x) = f(x; \hat\theta), \qquad \hat\theta = \underset{\theta}{\mathrm{argmin}} \ \frac{1}{|\mathcal{I}_1|} \sum_{i \in \mathcal{I}_1} \rho_\alpha (Y_i - f(X_i ; \theta))  + \mathcal{R}(\theta),
% \end{align}
% where $ \rho_\alpha $ is the pinball loss \cite{koenker1978regression,steinwart2011estimating}, defined by
% \begin{align} \label{eq:pinball}
% \rho_\alpha(z) := \begin{cases}
% \alpha z, & \text{if } z> 0, \\
% (\alpha-1)z ,& \text{otherwise};
% \end{cases}
% \end{align} 
% here, $\mathcal{R}(\cdot)$ is a possible regularizer (see description below). 
% In our experiments, $f(x;\theta)$ is a regression function formulated as a deep neural network with parameters $\theta$. The network design and training algorithm are 
identical to those described in \cite{romano2019conformalized} (once again, the source code is available online at {\url{https://github.com/yromano/cqr}}). Specifically, we use a two-hidden-layer neural network, with ReLU nonlinearities. The hidden dimension of both layers is set to 64. We use Adam optimizer \cite{kingma2014adam}, with minibatches of size $64$ and a fixed learning rate of $ 5 \times 10^{-4} $. We employ weight decay regularization with parameter equal to $ 10^{-6} $ and also use dropout \cite{srivastava2014dropout} with a dropping rate of $0.1$. We tune the number of epochs using cross validation (early stopping), with an upper limit of $1000$ epochs.

\color{teal}

\end{document}